\documentclass{llncs}

\usepackage{kc}

\title{Classical and Intuitionistic Subexponential Logics\\are Equally Expressive}
\titlerunning{CSL and ISL are equally expressive}
\author{Kaustuv Chaudhuri}
\institute{
  INRIA Saclay, France\\
  \email{kaustuv.chaudhuri@inria.fr}
}
\authorrunning{K. Chaudhuri}
\tocauthor{Kaustuv Chaudhuri (INRIA)}

\begin{document}

\maketitle

\pagestyle{plain}
\thispagestyle{plain}

\begin{abstract}
  It is standard to regard the intuitionistic restriction of a classical logic as increasing the
  expressivity of the logic because the classical logic can be adequately represented in the
  intuitionistic logic by double-negation, while the other direction has no truth-preserving
  propositional encodings. We show here that subexponential logic, which is a family of
  substructural refinements of classical logic, each parametric over a preorder over the
  subexponential connectives, does not suffer from this asymmetry if the preorder is systematically
  modified as part of the encoding. Precisely, we show a bijection between synthetic (i.e., focused)
  partial sequent derivations modulo a given encoding. Particular instances of our encoding for
  particular subexponential preorders give rise to both known and novel adequacy theorems for
  substructural logics.
\end{abstract}

\section{Introduction}
\label{sec:intro}

In~\cite{miller10bcs-my}, Miller writes:
\begin{quote}
  \itshape ``While there is some recognition that logic is a unifying and universal discipline
  underlying computer science, it is far more accurate to say that its universal character has been
  badly fractured \dots one wonders if there is any sense to insisting that there is a core notion
  of `logic'.''
\end{quote}
Possibly the oldest such split is along the classical/intuitionistic seam, and each side can be seen
as more universal than the other. Classical logics, the domain of traditional mathematics, generally
have an elegant symmetry in the connectives that can often be exploited to create sophisticated
proof search and model checking algorithms. On the other hand, intuitionistic logics, which
introduce an asymmetry between multiple hypotheses and single conclusions, can express the
computational notion of \emph{function} directly, making it the preferred choice for programming
languages and logical frameworks. Can the rift between these two sides be bridged?

Miller proposes one approach: to use structural proof theory, particularly the proof theory of
focused sequent calculi, as a unifying language for logical formalisms. There is an important proof
theoretic difference between a given classical logic and its \emph{intuitionistic restriction} (see
\defnref{int-res}): the classical formulas can be encoded using the intuitionistic connectives in
such a way that classical provability is preserved, \ie, a formula is classically provable if and
only if its encoding is intuitionistically provable. In the other direction, however, there are no
such general encodings. The classical logic will either have to be extended (for example, with terms
and quantifiers) or refined with substructural or modal operators. For this reason, intuitionistic
logics are sometimes considered to be \emph{more expressive} than their classical counterparts.

In this paper, we compare logical calculi for ``universality'' using the specific technical
apparatus of \emph{adequate propositional encodings}. That is, given a formula in a source logic
$O$, we must be able to encode it in a target logic $M$ that must preserve the atomic predicates and
must reuse the reasoning principles of $M$, particularly its notion of provability. An example of
such an encoding would be ordinary classical logic encoded in ordinary intuitionistic logic where
each classical formula $A$ is encoded as the intuitionistic formula $\lnot\lnot A$. We can go
further and also reuse the proofs of the target calculus; in fact, there are at least the following
\emph{levels} of adequacy:

\begin{defn}[levels of adequacy]
  An encoding of formulas (equiv. of sequents) from a source to a target calculus is
  \begin{itemize}
  \item \emph{globally adequate} if a formula is true (equiv. a sequent is derivable) in the source
    calculus iff its encoding is true (equiv. the encoding of the sequent is derivable) in the
    target calculus;
  \item \emph{adequate} if the proofs of a formula (equiv. a sequent) in the source calculus are in
    bijection with the proofs of the encoding of the formula (equiv. the sequent) in the target
    calculus; and
  \item \emph{locally adequate} if open derivations (\ie, partial proofs with possibly unproved
    premises) of a formula (equiv. a sequent) in the source calculus are in bijection with the open
    derivations of the formula (equiv. the sequent) in the target calculus.
  \end{itemize}
\end{defn}

Local adequacy is an ideal for encodings because it is a strong justification for seeing the target
calculus as more universal: (partial) proofs in the source calculus can be recovered at any level of
detail. However, it is unachievable except in trivial situations. Indeed, even adequacy is often
difficult; for instance, the linear formula $\llbang a \llimp \llbang b \llimp \llbang a$ has three
sequent proofs, differing in the order in which the second $\llimp$ and the two $\llbang$s are
introduced, but there is only a single sequent proof of $a \limp b \limp a$.

It is nevertheless possible to define a kind of local adequacy that is more flexible: adequacy up to
permutations of inference rules entirely inside one of the phases of \emph{focusing}. A focused
proof~\cite{andreoli92jlc} is a proof that makes large \emph{synthetic} rules that are maximal
chains of positive or negative inference rules. An inference rule is positive, sometimes called
synchronous, if it involves an essential choice, while it is negative or asynchronous if the choices
it presents (if any) are inessential. The term ``focus'' describes the way positive inferences are
chained to form synthetic steps: each inference is applied (read from conclusion to premises) to a
single formula \emph{under focus}, and the operands of this connective remain under focus in the
premises.
\begin{defn}[focal adequacy] \label{defn:focadq}
  An encoding of sequents from a source to a target focused calculus is \emph{focally adequate} if
  they have the same synthetic inference rules.
\end{defn}

Since focusing abstracts away the inessential permutations of inference rules, a focally adequate
encoding can be used to compare logics for ``essential universality''. Surprisingly, there are very
few known focal adequacy results (see \cite{chaudhuri06phd,liang09tcs} for practically all such
known results). This paper fills in many of the gaps for existing (substructural) logics by proving
a pair of general encodings (see theorems~\ref{thm:csl-isl} and \ref{thm:isl-csl}) about
\emph{subexponential} logics~\cite{danos93kgc-my,nigam09ppdp-my}. It is well known that the
exponentials of linear logic are non-canonical. If a pre-order is imposed upon them with suitable
conditions, then the resulting logic is well-behaved, satisfying identity, admitting cuts, and
allowing focusing. Moreover, classical, intuitionistic, and linear logics can be seen as
\emph{instances} of subexponential logic for particular collections of subexponentials. Our
encodings are \emph{generic}, parametric on the \emph{subexponential signature} of the source and
target logics. As particular instances, we obtain focal adequacy results for: classical logic (CL)
in intuitionistic logic (IL), IL in classical linear logic (CLL), CLL in intuitionistic linear logic
(ILL), and an indefinite bidirectional chain between classical and intuitionistic subexponential
logics, all of which are novel. Moreover, our encodings show that any analysis (such as
cut-elimination) or algorithm (such as proof search) that is generic on the subexponential signature
cannot (and \emph{need not}) distinguish between classical and intuitionistic logics.

The rest of this paper is organized as follows: in \secref{csl} classical subexponential logic is
introduced, together with its focused sequent calculus and well known instances; in \secref{isl} its
intuitionistic restriction is presented; then in \secref{encodings} the bidirectional encoding
between classical and intuitionistic subexponential logic is constructed. Details omitted here for
space reasons can be found in the accompanying technical report~\cite{chaudhuri10tr}.

\section{Classical subexponential logic}
\label{sec:csl}

Subexponential logic borrows most of its syntax from linear logic~\cite{girard87tcs}. As we are
comparing focused systems, we adopt a polarised syntax from the beginning. Polarised formulas will
have exactly one of two polarities: \emph{positive} ($P, Q, \ldots$) constructed out of the positive
atoms and connectives, and \emph{negative} ($N, M, \ldots$) constructed out of the negative atoms
and connectives. These two classes of formulas are mutually recursive, mediated by the indexed
subexponential operators $\llbang_z$ and $\llqmark_z$.

\begin{notn}[syntax]
  \emph{Positive formulas} ($P, Q$) and \emph{negative formulas} ($N,
  M$) have the following grammar:
  \begin{kcalign}
    P, Q &\Coloneqq p \gor P \lltens Q \gor \llone \gor P \llplus Q \gor \llzero \gor \llbang_z \pat N \tag*{(positive)} \\
    N, M &\Coloneqq n \gor N \llwith M \gor \lltop \gor N \llpar M
    \gor \llbot \gor P \llimp N \gor \llqmark_z \nat P \tag*{(negative)}
  \end{kcalign}
  Atomic formulas are written in lower case ($a, b, \ldots$), with $p$ and $q$ reserved for positive
  and $n$ and $m$ reserved for negative atomic formulas. $\nat P$ denotes either a positive formula
  or a negative atom, and likewise $\pat N$ denotes a negative formula or a positive atom. We write
  $A, B, \ldots$ for any arbitrary formula (positive or negative).
\end{notn}
Because we will eventually consider its intuitionistic restriction, we retain implication $\llimp$
as a primitive even though it is classically definable. However, we exclude the non-linear
implication ($\limp$) because the unrestricted zones are non-canonical; \ie, there are many such
implications, each defined using a suitable subexponential (or compositions thereof). The subscript
$z$ in exponential connectives denotes zones drawn from a \emph{subexponential signature} (using the
terminology of~\cite{nigam09ppdp-my}).

\begin{defn}
  A \emph{subexponential signature} $\Sigma$ is a structure $\langle
  Z, \le, \lin, U \rangle$ where:
  \begin{itemize}
  \item $\langle Z, \le \rangle$ is a non-empty pre-ordered set (the ``zones'');
  \item $\lin \in Z$ is a ``\emph{working}'' zone;
  \item $U \subseteq Z$ is a set of \emph{unrestricted} zones that is
    $\le$-closed, \ie, for every $z_1, z_2 \in Z$, if $z_1 \le z_2$,
    then $z_1 \in U$ implies $z_2 \in U$. $Z \setminus U$ will be
    called the \emph{restricted} zones.
  \end{itemize}
  We use $u, v, w$ to denote unrestricted zones and $r, s, t$ to
  denote restricted zones.
\end{defn}

Unrestricted zones admit both weakening and contraction, while restricted zones are linear. The
logic is parametric on the signature. (Particular mentions of the signature will be omitted unless
necessary to disambiguate, in which case they will be written in a subscript.)
We use use a two-sided sequent calculus formulation of the logic in order to avoid appeals to De
Morgan duality. This will not only simplify the definition of the intuitionistic restriction
(\secref{isl}), but will also be crucial to the main adequacy result. Formulas in contexts are
annotated with their subexponential zones as follows: $\zat{z:A}$ will stand for $A$ occurring in
zone denoted by $z$, and $\zat{z:(A_1, \ldots, A_k)}$ for $\zat{z:A_1}, \ldots, \zat{z:A_k}$.
Sequents are of the following kinds:

\smallskip
\begin{tabular}{r@{\ $\vdash$\ }l@{\qquad}l}
  $\zl$ & $\foc{P} \zs \zr$ & right focus on $P$ \\
  $\zl \zs \foc{N}$ & $\zr$ & left focus on $N$ \\
  $\zl \zs \zla$ & $\zra \zs \zr$ & active on $\zla$ and $\zra$
\end{tabular}
\smallskip

\noindent
The contexts in these sequents have the following restrictions:
\begin{itemize}
\item All elements of the \emph{left passive} context $\zl$ are of the
  form $\zat{z:\pat N}$.
\item All elements of the \emph{right passive} context $\zr$ are of
  the form $\zat{z:\nat P}$.
\item All elements of the \emph{left active} context $\zla$ are
   of the form $\nat P$.
\item All elements of the \emph{right active} context $\zra$ are
  of the form $\pat N$.
\end{itemize}

\begin{notn}
  We write $\unr\zl$ or $\unr\zr$ for those contexts containing only
  unrestricted elements, \ie, each element is of the form $\zat{u:A}$
  with $u \in U$. Likewise, we write $\res\zl$ or $\res\zr$ for
  contexts containing only restricted elements.
\end{notn}

\begin{figure}[htp]
  \noindent (right focus)
  \begin{kcgather}
    \linfer[\rname{pr}]{
      \unr\zl, \zat{z:p} \vdash \foc{p} \zs \unr\zr
    }{ }
    \quad
    \linfer[\rr{\lltens}]{
      \unr\zl, \res\zl_1, \res\zl_2 \vdash \foc{P \lltens Q} \zs \unr\zr, \res\zr_1, \res\zr_2
    }{
      \unr\zl, \res\zl_1 \vdash \foc{P} \zs \unr\zr, \res\zr_1
      &
      \unr\zl, \res\zl_2 \vdash \foc{Q} \zs \unr\zr, \res\zr_2
    }
    \quad
    \linfer[\rr{\llone}]{
      \unr\zl \vdash \foc{\llone} \zs \unr\zr
    }{ }
    \\[1ex]
    \linfer[\rr{\llplus}_i]{
      \zl \vdash \foc{P_1 \llplus P_2} \zs \zr
    }{
      \zl \vdash \foc{P_i} \zs \zr
    }
    \quad
    \infer[\rr{\llbang_z}]{
      \zl \vdash \foc{\llbang_z \pat N} \zs \zr
    }{
      \zl \zs \cdot \vdash \pat N \zs \zr
      &
      \bigl(\all\zat{x:A}\in \zl, \zr. z \le x\bigr)
    }
  \end{kcgather}
  \noindent (left focus)
  \begin{kcgather}
    \linfer[\rname{nl}]{
      \unr\zl \zs \foc{n} \vdash \unr\zr, \zat{z:n}
    }{}
    \quad
    \linfer[\lr{\llwith}_i]{
      \zl \zs \foc{P_1 \llwith P_2} \vdash \zr
    }{
      \zl \zs \foc{P_i} \vdash \zr
    }
    \quad
    \infer[\lr{\llpar}]{
      \unr\zl, \res\zl_1, \res\zl_2 \zs \foc{N \llpar M} \vdash \unr\zr, \res\zr_1, \res\zr_2
    }{
      \unr\zl, \res\zl_1 \zs \foc{N} \vdash \unr\zr, \res\zr_1
      &
      \unr\zl, \res\zl_2 \zs \foc{M} \vdash \unr\zr, \res\zr_2
    }
    \\[1ex]
    \infer[\lr{\llbot}]{
      \unr\zl \zs \foc{\llbot} \vdash \unr\zr
    }{ }
    \quad
    \infer[\lr{\llimp}]{
      \unr\zl, \res\zl_1, \res\zl_2 \zs \foc{P \llimp M} \vdash \unr\zr, \res\zr_1, \res\zr_2
    }{
      \unr\zl, \res\zl_1 \vdash \foc{P} \zs \unr\zr, \res\zr_1
      &
      \unr\zl, \res\zl_2 \zs \foc{M} \vdash \unr\zr, \res\zr_2
    }
    \\[1ex]
    \infer[\lr{\llqmark_z}]{
      \zl \zs \foc{\llqmark_z \nat P} \vdash \zr
    }{
      \zl \zs \nat P \vdash \cdot \zs \zr
      &
      \bigl(\all\zat{x:A}\in \zl, \zr. z \le x\bigr)
    }
  \end{kcgather}
  \noindent (right active)
  \begin{kcgather}
    \linfer[\rname{ar}]{
      \zl \zs \zla \vdash \zra, a \zs \zr
    }{
      \zl \zs \zla \vdash \zra \zs \zr, \zat{\lin:a}
    }
    \quad
    \linfer[\rr{\llwith}]{
      \zl \zs \zla \vdash \zra, N \llwith M \zs \zr
    }{
      \zl \zs \zla \vdash \zra, N \zs \zr
      &
      \zl \zs \zla \vdash \zra, M \zs \zr
    }
    \quad
    \linfer[\rr{\lltop}]{
      \zl \zs \zla \vdash \zra, \lltop \zs \zr
    }{ }
    \\[1ex]
    \linfer[\rr{\llpar}]{
      \zl \zs \zla \vdash \zra, N \llpar M \zs \zr
    }{
      \zl \zs \zla \vdash \zra, N, M \zs \zr
    }
    \quad
    \linfer[\rr{\llbot}]{
      \zl \zs \zla \vdash \zra, \llbot \zs \zr
    }{
      \zl \zs \zla \vdash \zra \zs \zr
    }
    \quad
    \linfer[\rr{\llimp}]{
      \zl \zs \zla \vdash \zra, P \llimp N \zs \zr
    }{
      \zl \zs \zla, P \vdash \zra, N \zs \zr
    }
    \quad
    \linfer[\rr{\llqmark_z}]{
      \zl \zs \zla \vdash \zra, \llqmark_z \nat P \zs \zr
    }{
      \zl \zs \zla \vdash \zra \zs \zr, \zat{z:\nat P}
    }
  \end{kcgather}
  \noindent (left active)
  \begin{kcgather}
    \linfer[\rname{al}]{
      \zl \zs \zla, a \vdash \zra \zs \zr
    }{
      \zl, \zat{\lin:a} \zs \zla \vdash \zra \zs \zr
    }
    \quad
    \linfer[\lr{\lltens}]{
      \zl \zs \zla, P \lltens Q \vdash \zra \zs \zr
    }{
      \zl \zs \zla, P, Q \vdash \zra \zs \zr
    }
    \quad
    \linfer[\lr{\llone}]{
      \zl \zs \zla, \llone \vdash \zra \zs \zr
    }{
      \zl \zs \zla \vdash \zra \zs \zr
    }
    \\[1ex]
    \linfer[\lr{\llplus}]{
      \zl \zs \zla, P \llplus Q \vdash \zra \zs \zr
    }{
      \zl \zs \zla, P \vdash \zra \zs \zr
      &
      \zl \zs \zla, Q \vdash \zra \zs \zr
    }
    \quad
    \linfer[\lr{\llzero}]{
      \zl \zs \zla, \llzero \vdash \zra \zs \zr
    }{ }
    \quad
    \linfer[\lr{\llbang_z}]{
      \zl \zs \zla, \llbang_z \pat N \vdash \zra \zs \zr
    }{
      \zl, \zat{z:\pat N} \zs \zla \vdash \zra \zs \zr
    }
  \end{kcgather}
  \noindent (decision)
  \begin{kcgather}
    \infer[\rname{rdr}]{
      \zl \zs \cdot \vdash \cdot \zs \zr, \zat{r:P}
    }{
      \zl \vdash \foc{P} \zs \zr
    }
    \quad
    \infer[\rname{udr}]{
      \zl \zs \cdot \vdash \cdot \zs \zr, \zat{u:P}
    }{
      \zl \vdash \foc{P} \zs \zr, \zat{u:P}
    }
    \quad
    \infer[\rname{rdl}]{
      \zl, \zat{r:N} \zs \cdot \vdash \cdot \zs \zr
    }{
      \zl \zs \foc{N} \vdash \zr
    }
    \quad
    \infer[\rname{udl}]{
      \zl, \zat{u:N} \zs \cdot \vdash \cdot \zs \zr
    }{
      \zl, \zat{u:N} \zs \foc{N} \vdash \zr
    }
  \end{kcgather}
  \caption{Focused sequent calculus for classical subexponential logic}
  \label{fig:csl-rules}
\end{figure}

The rules of the calculus are presented in \figref{csl-rules}. Focused
sequent calculi presented in this style, which is a stylistic variant
of Andreoli's original formulation~\cite{andreoli92jlc}, have an
intensional reading in terms of \emph{phases}. At the boundaries of
phases are sequents of the form $\zl \zs \cdot \vdash \cdot \zs \zr$,
which are known as \emph{neutral sequents}. Proofs of neutral sequents
proceed (reading from conclusion to premises) as follows:
\begin{enumerate} \parskip0.5\baselineskip
\item \emph{Decision}: a \emph{focus} is selected from a neutral
  sequent, either from the left or the right context. This focused
  formula is moved to its corresponding focused zone using one of the
  rules \rname{rdr}, \rname{udr}, \rname{rdl} and \rname{udl}
  (\rname{u}/\rname{r} = ``unrestricted''/``restricted'', \rname{d} =
  ``decision'', and \rname{r}/\rname{l} = ``right''/``left''). These
  \emph{decision} rules copy the focused formula iff it occurs in an
  unrestricted zone.

\item \emph{Focused phase}: for a left or a right focused sequent, left or right focus rules are
  applied to the formula under focus. These focused rules are all non-invertible in the (unfocused)
  sequent calculus and therefore depend on essential choices made in the proof. In all cases except
  \rr{\llbang_z} and \lr{\llqmark_z} the focus persists to the subformulas (if any) of the focused
  formula. For binary rules, the restricted portions of the contexts are separated and distributed
  to the two premises. This much should be familiar from focusing for linear
  logic~\cite{andreoli92jlc,chaudhuri08jar}.

  The two unusual rules for subexponential logic are \rr{\llbang_z} and \lr{\llqmark_z}, which are
  generalizations of rules for the single exponential in ordinary linear logic. These rules have a
  side condition that no formulas in a strictly $\le$-smaller zone may be present in the conclusion.
  If the working zone $\lin$ is $\le$-minimal (which is not necessarily the case), then this side
  condition is trivial and the rules amount to a pure change of polarities, similar to the
  $\uparrow$ and $\downarrow$ connectives of polarised linear logic~\cite{laurent02phd}. For the
  other zones, this rule tests for the emptiness of some of the zones. It is this selective
  emptiness test that gives subexponential logic its expressive
  power~\cite{nigam09ppdp-my,nigam09phd}.

\item \emph{Active phase}: once the exponential rules \rr{\llbang_z} and \lr{\llqmark_z} are
  applied, the sequents become active and left and right active rules are applied. The order of the
  active rules is immaterial as all orderings will produce the same list of neutral sequent
  premises. In Andreoli's system the irrelevant non-determinism in the order of these rules was
  removed by treating the active contexts $\zra$ and $\zla$ as ordered contexts; however, we do not
  fix any particular ordering.
\end{enumerate}

In the traditional model of focusing, the above three steps repeat, in that order, in the entire
proof. The focused system can therefore be seen as a system of \emph{synthetic} inference rules
(sometimes known as \emph{bipoles}) for neutral sequents. It is possible to give a very general
presentation of such synthetic inference systems, for which we can prove completeness and
cut-elimination in a very general fashion~\cite{chaudhuri08lpar-my}. It is also possible, with some
non-trivial effort, to show completeness of the focused calculus without appealing to synthetic
rules~\cite{chaudhuri08jar,liang09tcs}. We do not delve into such proofs in this paper because this
ground is well trodden. Indeed, a focused completeness theorem for a very similar (but more general)
formulation of subexponential logic can be found in~\cite[chapter 6]{nigam09phd}. The synthetic
soundness and completeness theorems are as follows, proof omitted:

\begin{fac}[synthetic soundness and completeness]
  Write $\Vdash$ for the sequent arrow for an unfocused variant of the calculus of
  \figref{csl-rules}, obtained by placing the focused and active formulas in the $\lin$ zone and
  relaxing the focusing discipline.\footnote{This is basically Gentzen's LK in two-sided form for
    subexponential logic.}
  \begin{enumerate}
  \item If $\zl \zs \cdot \vdash \cdot \zs \zr$, then $\zl \Vdash \zr$ (synthetic soundness).
  \item If $\zl, \zat{\lin:\zla} \Vdash \zat{\lin:\zra}, \zr$ then $\zl \zs \zla \vdash \zra \zs
    \zr$ (synthetic completeness). \qed
  \end{enumerate}
\end{fac}

Despite its somewhat esoteric formulation, it is easy to see how subexponential logic generalizes
classical substructural logics.

\begin{fac}[familiar instances] \label{defn:csl-fam} \mbox{}
  \begin{itemize}
  \item \emph{Polarised classical multiplicative additive linear logic} (MALL) is determined by
    $\mathtt{mall} = \< \{\lin\}, \cdot, \lin, \emptyset \>$. The injections between the two
    polarised classes, sometimes known as \emph{shifts}, are as follows: $\downarrow =
    \llbang_{\lin}$ and $\uparrow = \llqmark_\lin$.
  \item \emph{Polarised classical linear logic} (CLL) is determined by $\mathtt{ll} = \< \{\lin,
    \zunr\}, \lin \le \zunr, \lin, \{\zunr\} \>$. In addition to the injections of \texttt{mall}, we
    also have the exponentials $\llbang = \llbang_\zunr$ and $\llqmark = \llqmark_\zunr$.
  \item \emph{Polarised classical logic} (CL) is given by the signature $\mathtt{l} = \< \{\lin\},
    \cdot, \lin, \{\lin\}\>$. \qed
  \end{itemize}
\end{fac}

In addition to such instances produced by instantiating the subexponential signature, it is also
possible to get the unpolarised versions of these logics by applying $\llbang_\lin$ and
$\llqmark_\lin$ to immediate negative (resp. positive) subformulas of positive (resp. negative)
formulas.

\section{Intuitionistic subexponential logic}
\label{sec:isl}

One direct way of defining intuitionistic fragments of classical logics is as follows:

\begin{defn}[intuitionistic restriction] \label{defn:int-res}
  Given a two-sided sequent calculus, its \emph{intuitionistic restriction} is that fragment where
  all inference rules are constrained to have exactly a single formula on the right hand sides of
  sequents.
\end{defn}

The practical import of this restriction is that the connectives $\llpar$ and $\llbot$ disappear,
because their right rules require two and zero conclusions, respectively. As a result, $\llimp$
becomes a primitive because its classical definition requires $\llpar$ (and De Morgan duals, which
are also missing with the intuitioistic restriction). In a slight break from
tradition~\cite{girard87tcs,schellinx91jlc,barber96tr}, we retain $\llqmark_z$ in the intuitionistic
syntax.
The intuitionistic restriction produces the following kinds of sequents:

\smallskip
\begin{tabular}{r@{\ $\vdash$\ }l@{\qquad}l}
  $\zl$ & $\foc{P}$ & right focus on $P$ \\
  $\zl \zs \foc{N}$ & $\zat{z:\nat Q}$ & left focus on $N$ \\
  $\zl \zs \zla$ & $\pat N \zs \cdot$ & active on $\zla$ and $\pat N$ \\
  $\zl \zs \zla$ & $\cdot \zs \zat{z:\nat Q}$ & active on $\zla$
\end{tabular}
\smallskip

\noindent
We shall use $\zg$ to stand for the right hand forms---either $\pat N \zs \cdot$ or $\cdot \zs
\zat{z:\nat Q}$---for active sequents above. The full collection of rules is given in
\figref{isl-rules}. As before, we use $\nat Q$ (resp. $\pat N$) to refer to a positive formula or
negative atom (resp. negative formula or positive atom).

The nature of subexponential signatures does not change in moving from classical to intuitionistic
logic. The decision rule \rname{udr} obviously cannot copy the right formula in the intuitionistic
case. Thus, both the right decision rules collapse; $\llqmark_z$ takes on an additional modal aspect
and is no longer the perfect dual of $\llbang_z$. The standard explanation of this loss of symmetry
in the exponentials is the creation of a new \emph{possibility} judgement that is weaker than linear
truth; see~\cite{chaudhuri03tr} for such a reconstruction of the intuitionistic $\llqmark$.

\begin{figure}[htp]
  \noindent (right focus)
  \begin{kcgather}
    \linfer[\rname{pr}]{
      \unr\zl, \zat{z:p} \vdash \foc{p}
    }{ }
    \quad
    \linfer[\rr{\lltens}]{
      \unr\zl, \res\zl_1, \res\zl_2 \vdash \foc{P \lltens Q}
    }{
      \unr\zl, \res\zl_1 \vdash \foc{P}
      &
      \unr\zl, \res\zl_2 \vdash \foc{Q}
    }
    \quad
    \linfer[\rr{\llone}]{
      \unr\zl \vdash \foc{\llone}
    }{ }
    \\[1ex]
    \linfer[\rr{\llplus}_i]{
      \zl \vdash \foc{P_1 \llplus P_2}
    }{
      \zl \vdash \foc{P_i}
    }
    \quad
    \infer[\rr{\llbang_z}]{
      \zl \vdash \foc{\llbang_z\pat N}
    }{
      \zl \zs \cdot \vdash \pat N \zs \cdot
      &
      \bigl(\all\zat{x:A}\in \zl. z \le x\bigr)
    }
  \end{kcgather}
  \noindent (left focus)
  \begin{kcgather}
    \linfer[\rname{nl}]{
      \unr\zl \zs \foc{n} \vdash \zat{z:n}
    }{}
    \quad
    \linfer[\lr{\llwith}_i]{
      \zl \zs \foc{P_1 \llwith P_2} \vdash \zat{z:\nat Q}
    }{
      \zl \zs \foc{P_i} \vdash \zat{z:\nat Q}
    }
    \quad
    \infer[\lr{\llimp}]{
      \unr\zl, \res\zl_1, \res\zl_2 \zs \foc{P \llimp M} \vdash \zat{z:\nat Q}
    }{
      \unr\zl, \res\zl_1 \vdash \foc{P}
      &
      \unr\zl, \res\zl_2 \zs \foc{M} \vdash \zat{z:\nat Q}
    }
    \\[1ex]
    \infer[\lr{\llqmark_z}]{
      \zl \zs \foc{\llqmark_z\nat P} \vdash \zat{y:\nat Q}
    }{
      \zl \zs \nat P \vdash \cdot \zs \zat{y:\nat Q}
      &
      \bigl(\all\zat{x:A}\in \zl, \zat{y:\nat Q}. z \le x\bigr)
    }
  \end{kcgather}
  \noindent right active
  \begin{kcgather}
    \linfer[\rname{ar}]{
      \zl \zs \zla \vdash a \zs \cdot
    }{
      \zl \zs \zla \vdash \cdot \zs \zat{\lin:a}
    }
    \quad
    \linfer[\rr{\llwith}]{
      \zl \zs \zla \vdash N \llwith M \zs \cdot
    }{
      \zl \zs \zla \vdash N \zs \cdot
      &
      \zl \zs \zla \vdash M \zs \cdot
    }
    \quad
    \linfer[\rr{\lltop}]{
      \zl \zs \zla \vdash \lltop \zs \cdot
    }{ }
    \\[1ex]
    \linfer[\rr{\llimp}]{
      \zl \zs \zla \vdash P \llimp N \zs \cdot
    }{
      \zl \zs \zla, P \vdash N \zs \cdot
    }
    \quad
    \linfer[\rr{\llqmark_z}]{
      \zl \zs \zla \vdash  \llqmark_z P \zs \cdot
    }{
      \zl \zs \zla \vdash \cdot \zs \zat{z:P}
    }
  \end{kcgather}
  \noindent (left active)
  \begin{kcgather}
    \linfer[\rname{al}]{
      \zl \zs \zla, a \vdash \zg
    }{
      \zl, \zat{\lin:a} \zs \zla \vdash \zg
    }
    \quad
    \linfer[\lr{\lltens}]{
      \zl \zs \zla, P \lltens Q \vdash \zg
    }{
      \zl \zs \zla, P, Q \vdash \zg
    }
    \quad
    \linfer[\lr{\llone}]{
      \zl \zs \zla, \llone \vdash \zg
    }{
      \zl \zs \zla \vdash \zg
    }
    \\[1ex]
    \linfer[\lr{\llplus}]{
      \zl \zs \zla, P \llplus Q \vdash \zg
    }{
      \zl \zs \zla, P \vdash \zg
      &
      \zl \zs \zla, Q \vdash \zg
    }
    \quad
    \linfer[\lr{\llzero}]{
      \zl \zs \zla, \llzero \vdash \zg
    }{ }
    \quad
    \linfer[\lr{\llbang_z}]{
      \zl \zs \zla, \llbang_z N \vdash \zg
    }{
      \zl, \zat{z:N} \zs \zla \vdash \zg
    }
  \end{kcgather}
  \noindent (decision)
  \begin{kcgather}
    \infer[\rname{dr}]{
      \zl \zs \cdot \vdash \cdot \zs \zat{z:P}
    }{
      \zl \vdash \foc{P}
    }
    \quad
    \infer[\rname{rdl}]{
      \zl, \zat{r:N} \zs \cdot \vdash \cdot \zs \zat{z:\nat Q}
    }{
      \zl \zs \foc{N} \vdash \zat{z:\nat Q}
    }
    \quad
    \infer[\rname{udl}]{
      \zl, \zat{u:N} \zs \cdot \vdash \cdot \zs \zat{z:\nat Q}
    }{
      \zl, \zat{u:N} \zs \foc{N} \vdash \zat{z:\nat Q}
    }
  \end{kcgather}
  \caption{Focused sequent calculus for intuitionstic subexponential
    logic}
  \label{fig:isl-rules}
\end{figure}

The proof of completeness for focused intuitionistic subexponential logic has never been published.
However, any similar proof for intuitionistic linear logic, such
as~\cite{chaudhuri08jar,liang09tcs}, can be adapted. Again, we simply state the synthetic version of
the theorems here without proof.

\begin{fac}[synthetic soundness and completeness]
  Write $\Vdash$ for the sequent arrow for an unfocused variant of the calculus of
  \figref{isl-rules}, obtained by placing the focused and active formulas in the $\lin$ zone and
  relaxing the focusing discipline.
  \begin{enumerate}
  \item If $\zl \zs \cdot \vdash \cdot \zs \zat{z:\nat Q}$, then $\zl \Vdash \zat{z:\nat Q}$.
  \item If $\zl, \zat{\lin:\zla} \Vdash \zat{z:\nat Q}$ then $\zl \zs \zla \vdash \cdot \zs
    \zat{z:\nat Q}$.
  \item If $\zl, \zat{\lin:\zla} \Vdash \zat{\lin:N}$ then $\zl \zs \zla \vdash N \zs \cdot$. \qed
  \end{enumerate}
\end{fac}
The intuitionstic restrictions of the familiar instances from \defnref{csl-fam} simply use the same
subexponential signatures.


\section{Focally adequate encodings}
\label{sec:encodings}

This section contains the main technical contribution of this paper: focally adequate encodings
(\defnref{focadq}) that are generic on subexponential signatures. At the level of focal adequacy,
therefore, the asymmetry in the expressive power of classical and intuitionistic logics disappears.

\subsection{Classical in intuitionistic}
\label{sec:encodings:csl-isl}

To introduce the mechanisms of encoding, we first look at the unsurprising direction: a classical
logic in its own intuitionistic restriction. The well known double negation translation, if
performed clumsily, can break even full adequacy. For example, if $N \llpar M$ is translated as
$\lnot (\llbang_\lin \lnot \llbang_\lin N \lltens \llbang_\lin \lnot \llbang_\lin M)$ where $\lnot P
\triangleq P \llimp k$ where $k$ is some fixed negative atom that is not used in classical logic. In
the rule \rr{\llpar} under this encoding, there are instances of $\llbang_\lin$ that have no
counterpart in the classical side.
Indeed, there is no derived rule in the classical focused calculus that allows one to conclude $\zl
\zs \cdot \vdash N \llpar M \zs \cdot$ from $\zl \zs \cdot \vdash \cdot \zs \llbang_\lin N,
\llbang_\lin M$, which is what would result if the active phase could be suspended arbitrarily and
the subformula property were discarded. Such a rule is certainly admissible, but admissibile rules
do not preserve bijections between proofs, and are only definable for full proofs in any case.

How does one encode classical logic in its intuitionistic restriction such that polarities are
respected? The above example suggests an obvious answer: when translating $N \llpar M$ as if it were
right-active, do not also translate the subformulas $M$ and $N$ as if they were right-active, for
they will be sent to the left. Instead, translate them as if they were
\emph{left}-active.\footnote{The astute reader might recall that this is the essence of Kuroda's
  encodings.}

\begin{defn}[encoding classical formulas] \label{defn:trans-csl-isl}\mbox{}
  \begin{itemize}
  \item The encoding \teq{-} from classical positive (resp. negative) formulas to intuitionistic
    positive (resp. negative) formulas is as follows:
    \begin{kcalign}
      \teq{p} &= p &
      \teq{\llbang_z N} &= \llbang_z\teq N &
      \teq{P \lltens Q} &= \teq{P} \lltens \teq{Q} &
      \teq{\llone} &= \llone \\
      \teq{P \llplus Q} &= \teq{P} \llplus \teq{Q} &
      \teq{\llzero} &= \llzero &
      \teq{N} &= \lnot \tne{N}
    \end{kcalign}

  \item The encoding \tne{-} from classical negative (resp. positive)
    formulas to intuitionstic positive (resp. negative) formulas is as follows:
    \begin{kcalign}
      \tne{n} &= \dmdual n &
      \tne{\llqmark_z P} &= \llbang_z\tne P &
      \tne{N \llpar N} &= \tne{N} \lltens \tne{M} &
      \tne{\llbot} &= \llone \\
      \tne{N \llwith M} &= \tne{N} \llplus \tne{M} &
      \tne{\lltop} &= \llzero &
      \tne{P \llimp N} &= \teq{P} \lltens \tne{N} &
      \tne{P} &= \lnot\teq{P}
    \end{kcalign}
    where for every negative atom $n$, there is a positive atom $\dmdual n$ in the encoding.
  \end{itemize}
\end{defn}
Contexts are translated element-wise.

\begin{defn}[encoding classical sequents] \label{defn:csl-isl} \small
  The encoding $\dblneg{-}$ of classical sequents as intuitionistic
  sequents is as follows:
  \begin{kcalign}
    \dblneg{\zl \vdash \foc{P} \zs \zr}
    &=
    \teq\zl, \tne\zr \vdash \foc{\teq P}
    \qquad
    \dblneg{\zl \zs \foc{N} \vdash \zr}
    =
    \teq\zl, \tne\zr \vdash \foc{\smash{\tne N}}
    \\
    \dblneg{\zl \zs \zla \vdash \zra \zs \zr}
    &=
    \teq\zl, \tne\zr \zs \teq\zla,\tne\zra \vdash \cdot \zs \zat{\lin:k}
  \end{kcalign}
\end{defn}

In other words, focused sequents are translated to right-focused sequents, and active sequents to
left-active sequents. The right contexts are dualised and sent to the left where the intuitionistic
restriction does not apply, while the left focus on negative formulas is turned into a right focus
because of the lack of a multiplicative left-focused rule (except $\lr{\llimp}$ which would cause an
inadvertent polarity switch).

\begin{thm} \label{thm:csl-isl}
  The encoding of \defnref{csl-isl} is focally adequate
  (\defnref{focadq}).
\end{thm}

\begin{proof}
  We will inventory the classical rules in \figref{csl-rules}, and in each case compute the
  intuitionistic synthetic derivations of the encoding of the conclusion of the classical rules.
  Here are the interesting\footnote{See \cite{chaudhuri10tr} for the remaining cases.} cases, with
  the double inference lines denoting (un)folding of defns.~\ref{defn:trans-csl-isl} and
  \ref{defn:csl-isl}, and the rule names written with the prefix \rname{c/} or \rname{i/} to
  distinguish between classical and intuitionistic respectively.
  \begin{itemize}
  \item \emph{cases of \rname{c/pr} and \crr{\llbang}}:
    \begin{kcgather}
      \infer={
        \dblneg{\unr\zl, \zat{z:p} \vdash \foc{p} \zs \unr\zr}
      }{
        \infer={
          \teq{\unr\zl}, \teq{\zat{z:p}}, \tne{\unr\zr} \vdash \foc{p}
        }{
          \infer[\rname{i/pr}]{
            \teq{\unr\zl}, \zat{z:p}, \tne{\unr\zr} \vdash \foc{p}
          }{}
        }
      }
      \qquad
      \infer={
        \dblneg{\zl \vdash \foc{\llbang_z N} \zs \zr}
      }{
        \infer={
          \teq\zl, \tne\zr \vdash \foc{\teq{\llbang_z N}}
        }{
          \infer[\irr{\llbang_z}]{
            \teq\zl, \tne\zr \vdash \foc{\llbang_z\lnot\tne{N}}
          }{
            \infer[\irr{\llimp}]{
              \teq\zl, \tne\zr \zs \cdot \vdash \lnot\tne{N} \zs \cdot
            }{
              \infer[\fbox{\irr{\llqmark_\lin}}]{
                \teq\zl, \tne\zr \zs \tne{N} \vdash k \zs \cdot
              }{
                \infer={
                  \teq\zl, \tne\zr \zs \tne{N} \vdash \cdot \zs \zat{\lin:k}
                }{
                  \dblneg{\zl \zs \cdot \vdash N \zs \zr}
                }
              }
            }
          }
        }
      }
    \end{kcgather}
    All the logical rules used are invertible. The boxed instance of \irr{\llqmark_\lin} requires
    some explanation: obviously a left active rule on \tne{N} can be applied before this rule.
    However, since they are both active rules, the choice of which to perform first is immaterial as
    they will produce the same neutral premises. If we want local---not focal---adequacy, we will
    have to impose a right-to-left ordering on the active rules. The case of \rname{c/nl} and
    \clr{\llqmark} is similar.
  \item \emph{case of \crr{\llpar}}:
    \begin{kcgather}
      \linfer={
        \dblneg{\zl \zs \zla \vdash \zra, N \llpar M \zs \zr}
      }{
        \infer={
          \teq\zl, \tne\zr \zs \teq\zla, \tne\zra, \tne{N \llpar M} \vdash \cdot \zs \zat{\lin:k}
        }{
          \infer[\ilr{\lltens}]{
            \teq\zl, \tne\zr \zs \teq\zla, \tne\zra, \tne{N} \lltens \tne{M} \vdash \cdot \zs \zat{\lin:k}
          }{
            \infer={
              \teq\zl, \tne\zr \zs \teq\zla, \tne\zra, \tne{N}, \tne{M} \vdash \cdot \zs \zat{\lin:k}
            }{
              \dblneg{\zl \zs \zla \vdash \zra, N, M \zs \zr}
            }
          }
        }
      }
    \end{kcgather}
    The cases of \crr{\llbot}, \clr{\llbang_z} and \crr{\llqmark_z}
    are similar.

  \item \emph{case of \rname{c/rdr}}:
    \begin{kcgather}
      \infer={
        \dblneg{{\unr\zl_1, \res\zl_2} \zs \cdot \vdash \cdot \zs {\unr\zr_1, \res\zr_2}, \zat{r:P}}
      }{
        \infer={
          \teq{\unr\zl_1, \res\zl_2}, \tne{{\unr\zr_1, \res\zr_2}}, \tne{\zat{r:P}} \zs \cdot \vdash \cdot \zs \zat{\lin:k}
        }{
          \infer[\rname{i/rdl}]{
            \teq{\unr\zl_1, \res\zl_2}, \tne{{\unr\zr_1, \res\zr_2}}, \zat{r:\lnot\teq P} \zs \cdot \vdash \cdot \zs \zat{\lin:k}
          }{
            \infer[\ilr{\llimp}]{
              \teq{\unr\zl_1, \res\zl_2}, \tne{{\unr\zr_1, \res\zr_2}} \zs \foc{\lnot\teq P} \vdash \zat{\lin:k}
            }{
              \infer{
                \teq{\unr\zl_1, \res\zl_2}, \tne{{\unr\zr_1, \res\zr_2}} \vdash \foc{\teq P}
              }{
                \dblneg{{\unr\zl_1, \res\zl_2} \vdash \foc{P} \zs {\unr\zr_1, \res\zr_2}}
              }
              &
              \infer[\rname{i/nl}]{
                \teq{\unr\zl_1}, \tne{\unr\zr_1} \zs \foc{k} \vdash \zat{\lin:k}
              }{ }
            }
          }
        }
      }
    \end{kcgather}
    Note that the right premise is forced to terminate in the same phase. This would not be possible
    if, instead of $k$, we were to use some other negative formula such as $\llqmark_\lin \llzero$.
    In the presence of some unrestricted subexponential $u$, we might have used $\llqmark_u \llzero$
    instead (note that, classically, $\llqmark_u \llzero \equiv \llbot$). \qed
  \end{itemize}
\end{proof}


\begin{cor} \mbox{}
  \begin{itemize}
  \item There is a focally adequate encoding of classical MALL in
    intuitionistic MALL.
  \item There is a focally adequate encoding of CLL in ILL.
  \item There is a focally adequate encoding of CL in IL.
  \end{itemize}
\end{cor}

\begin{proof}
  Instantiate \thmref{csl-isl} on the subexponential signatures from
  \defnref{csl-fam}. \qed
\end{proof}

\noindent
These instances are all apparently novel, partly because focal adequacy of classical logics in their
own intuitionistic restrictions has not been deeply investigated. In the work on
LJF~\cite{liang09tcs} there is a focally adequate encoding of classical logic in intuitionistic
linear logic, which can be seen as a combination of the second and third of the above instances.

\subsection{Intuitionistic in classical}
\label{sec:encodings:isl-jsl}

The previous subsection showed that the intuitionistic restriction of a classical logic can
adequately encode the classical logic itself. This is not the case in the other direction without
further modifications to the subexponential signature. It is easy to see this: consider just the
MALL fragment and the problem of encoding the \ilr{\llimp} rule. If $\llimp$ is encoded as itself,
then in the classical side we have the following derived rule (all the zones are $\lin$, and
elided):
\begin{kcgather}
  \infer{
    \zl \zs \foc{P \llimp N} \vdash \nat Q
  }{
    \zl \vdash \foc{P} \zs \nat Q
    &
    \zl \zs \foc{N} \vdash \cdot
  }
\end{kcgather}
This rule has no intuitionistic counterpart. Therefore, the encoding of $\llimp$ must prevent the
right formula $\nat Q$ from being sent to the left branch, \ie, to test that the rest of the right
context in a right focus is empty. MALL itself cannot perform this test because it lacks any truly
modal operators. Exactly the same problem exists for the encoding of IL in CL, which also lacks any
true modal operators.

Quite obviously, the encoding of $\llimp$ requires some means of testing the emptiness of contexts.
CLL (\defnref{csl-fam}) has an additional zone $\zunr$ that is greater than $\lin$, and therefore
$\llbang_\zunr$ can test for the absence of any $\lin$-formulas. It turns out that this is enough to
get a focally adequate encoding of IL as follows: the sole zone $\lin$ of IL is split into two,
$\lin_r$ (restricted) and $\lin_u$ (unrestricted), and the right hand side of IL sequents is encoded
with $\lin_r$.
Then, whenever $P$ is of the form $\llbang_\lin N$, the translation of it on the right is of the
form $\llbang_{\lin_u} M$.
In the rest of this subsection, we will systematically extend this observation to an arbitrary
subexponential signature.

\def\lj{\mathtt{l}}
\def\rj{\mathtt{r}}
\def\hatle{\mathbin{\hat\le}}
\def\nhatle{\mathbin{\hat\nleq}}

\begin{defn}[signature splitting] \label{defn:sig-split}
  Let a subexponential signature $\Sigma = \<Z, \le, \lin, U\>$ be given. Write:
  \begin{itemize}
  \item $\hat Z$ for the zone set $(Z \times \{\lj\}) \cup (Z \times \{\rj\})$, where $\lj$ and
    $\rj$ are distinct labels for the left and the right of the sequents, respectively, and $\times$
    is the Cartesian product. $Z \times \{\lj\}$ will be called the \emph{left form} of $\hat Z$,
    and $Z \times \{\rj\}$ will be called its \emph{right form}.
  \item $\hat U$ for the unrestricted zone set $U \times \{\lj\}$.
  \item $\hat \lin$ for the working zone $(\lin, \lj)$.
  \item $\hatle$ for the smallest relation on $\hat Z \times \hat Z$
    for which:
    \begin{itemize}
    \item[--] $(x, \lj) \hatle (y, \lj)$ if $x \le y$;
    \item[--] $(x, \rj) \hatle (y, \rj)$ if $x \le y$; and
    \item[--] $(x, \rj) \hatle (x, \lj)$ and $(x, \lj) \nhatle (x, \rj)$.
    \end{itemize}
  \end{itemize}
  The subexponential signature $\hat\Sigma = \<\hat Z, \hatle,
  \hat\lin, \hat U\>$ will be called the \emph{split form} of
  $\Sigma$.
\end{defn}

We intend to treat the right form specially. The zones in the right form are restricted, which
encodes the linearity of the right hand side inherent in the intuitionistic restriction
(\defnref{int-res}).
Our encoding will guarantee that the right hand sides of sequents in the encoding contain no zones
in the left form. Thus, when $\llbang_{(z, \lj)} N$ is under right focus, the side condition on the
$\rr{\llbang}$ rule will ensure that there are no other formulas on the right hand side, because the
right forms are made pointwise smaller than their left forms. Dually, on the left we shall use
$\llqmark_{(z, \rj)}$ to encode $\llqmark_z$; since the right form zones are pointwise smaller than
the left form zones, but retain the pre-split ordering inside their own zone, the side conditions
enforce the same occurrences as in the source calculus.

\def\tlp#1{\left(#1\right)^{\mathtt{lp}}}
\def\trp#1{\left(#1\right)^{\mathtt{rp}}}
\def\tlf#1{\left(#1\right)^{\mathtt{lf}}}
\def\trf#1{\left(#1\right)^{\mathtt{rf}}}
\def\tla#1{\left(#1\right)^{\mathtt{la}}}
\def\tra#1{\left(#1\right)^{\mathtt{ra}}}

\begin{defn}[encoding intuitionistic contexts]
  \label{defn:trans-isl-csl} \mbox{}
  \begin{itemize}
  \item The left-passive context $\zl$ is encoded pointwise using the
    translation $\tlp{-}$:
    \begin{kcalign}
      \tlp{\zat{z:\pat N}} &= \zat{(z, \lj):\tlp{\pat N}} &
      \tlp{p} &= p &
      \tlp{N} &= \tlf{N}
    \end{kcalign}
  \item A left-focused formula $N$ is encoded using the translation
    $\tlf{-}$:
    \begin{kcalign}
      \tlf{n} &= n \qquad
      \tlf{\llqmark_z \nat P} = \llqmark_{(z, \rj)} \tla{\nat P} \quad
      \tlf{N \llwith M} = \tlf{N} \llwith \tlf{M} \quad
      \tlf{\lltop} = \lltop \\
      \tlf{P \llimp N} &= \trf{P} \llimp \tlf{N}
    \end{kcalign}
  \item A right-focused formula $P$ is encoded using the translation
    $\trf{-}$:
    \begin{kcalign}
      \trf{p} &= p \qquad
      \trf{\llbang_z\pat N} = \llbang_{(z,\lj)} \tra{\pat N} \quad
      \trf{P \lltens Q} = \trf{P} \lltens \trf{Q} \quad
      \trf{\llone} = \llone \\
      \trf{P \llplus Q} &= \trf{P} \llplus \trf{Q} \quad
      \trf{\llzero} = \llzero
    \end{kcalign}
  \item A left-active context $\zla$ is encoded pointwise using the
    translation $\tla{-}$:
    \begin{kcalign}
      \tla{a} &= \llbang_{(\lin, \lj)} a \qquad
      \tla{\llbang_z\pat N} = \llbang_{(z, \lj)} \tlp{\pat N} \quad
      \tla{P \lltens Q} = \tla{P} \lltens \tla{Q} \quad
      \tla{\llone} = \llone \\
      \tla{P \llplus Q} &= \tla{P} \llplus \tla{Q} \quad
      \tla{\llzero} = \llzero
    \end{kcalign}
  \item A right-active formula $\pat N$ is encoded using the
    translation $\tra{-}$:
    \begin{kcalign}
      \tra{a} &= \llbang_{(\lin, \rj)} a \qquad
      \tra{\llqmark_z\nat P} = \llqmark_{(z, \rj)} \trp{\nat P} \quad
      \tra{N \llwith M} = \tra{N} \llwith \tra{M} \quad
      \tra{\lltop} = \lltop \\
      \tra{P \llimp N} &= \tla{P} \llimp \tra{N}
    \end{kcalign}
  \item A right-passive zoned formula $\zat{z:\nat P}$ is encoded using the
    translation $\trp{-}$:
    \begin{kcalign}
      \trp{\zat{z:\nat P}} &= \zat{(z,\rj):\trp{\nat P}} &
      \trp{n} &= n &
      \trp{P} &= \trf{P}
    \end{kcalign}
  \end{itemize}
\end{defn}
The cases for $\trf{\llbang_z\pat N}$ and $\tlf{\llqmark_z\nat P}$
will be crucial for the proof of \thmref{isl-csl}. Most of the
remaining cases can be seen as an abstract interpretation of the
focused rules (\figref{isl-rules}) on the various contexts. The
definition of the encoding of intuitionistic sequents is now
completely systematic.

\begin{defn}[encoding intuitionistic sequents] \small
  \label{defn:isl-csl}
  The encoding $\eic{-}$ of intuitionistic sequents as classical
  sequents is as follows:
  \begin{kcalign}
    \eic{\zl \vdash \foc{P}}
    &=
    \tlp{\zl} \vdash \foc{\trf{P}} \zs \cdot &
    \eic{\zl \zs \foc{N} \vdash \zat{z:\nat Q}}
    &=
    \tlp{\zl} \zs \foc{\tlf{N}} \vdash \trp{\zat{z:\nat Q}} \\
    \eic{\zl \zs \zla \vdash \pat N \zs \cdot}
    &=
    \tlp{\zl} \zs \tla{\zra} \vdash \tra{\pat N} \zs \cdot &
    \eic{\zl \zs \zla \vdash \cdot \zs \zat{z:\nat Q}}
    &=
    \tlp{\zl} \zs \tla{\zra} \vdash \cdot \zs \trp{\zat{z:\nat Q}}
  \end{kcalign}
\end{defn}
Observe that the right hand sides of the encoding have the
intuitionistic restriction (\defnref{int-res}). This restriction will
be enforced at every transtion from a focused to an active phase,
which is enough because the active rules cannot increase the size of
the right contexts.

\begin{thm} \label{thm:isl-csl}
  The encoding of \defnref{isl-csl} is focally adequate
  (\defnref{focadq}).
\end{thm}

\begin{proof}
  As before for \thmref{csl-isl}, we shall prove this by inventorying
  the intuitionistic rules of \figref{isl-rules}, encode the
  conclusions of each of these rules, and observe whether the neutral
  premises of the derived inference rules are in bijection with those
  of the \figref{isl-rules}. All but the following important cases are
  omitted here for space reasons.\footnote{See \cite{chaudhuri10tr}.}
  \begin{itemize}
  \item \emph{cases of \rname{i/pr} and \irr{\llbang_z}}:
    \begin{kcgather}
      \infer={
        \eic{\unr\zl, \zat{z:p} \vdash \foc{p}}
      }{
        \infer={
          \tlp{\unr\zl}, \tlp{\zat{z:p}} \vdash \foc{\trf{p}}
        }{
          \infer[\rname{c/pr}]{
            \tlp{\unr\zl}, \zat{z:p} \vdash \foc{p}
          }{ }
        }
      }
      \qquad
      \infer={
        \eic{\zl \vdash \foc{\llbang_z\pat N}}
      }{
        \infer={
          \tlp\zl \vdash \foc{\trf{\llbang_z\pat N}} \zs \cdot
        }{
          \infer[\fbox{\crr{\llbang}}]{
            \tlp\zl \vdash \foc{\llbang_{(z,\lj)}\tla{\pat N}} \zs \cdot
          }{
            \infer={
              \tlp\zl \zs \cdot \vdash \tla{\pat N} \zs \cdot
            }{
              \eic{\zl \zs \cdot \vdash \pat N \zs \cdot}
            }
          }
        }
      }
    \end{kcgather}
    The boxed instance of \crr{\llbang} is valid because all the zoned
    formulas in $\tlp\zl$ are in the left form zones, as is the zone
    of the $\llbang$ itself, so the comparison $\hatle$ is the same as
    $\le$ on the intuitionistic zones (\defnref{sig-split}).

  \item \emph{case of \ilr{\llimp}}:
    \begin{kcgather}
      \infer={
        \eic{\unr\zl, \res\zl_1, \res\zl_2 \zs \foc{P \llimp N} \vdash \zat{z:\nat Q}}
      }{
        \infer={
          \tlp{\unr\zl, \res\zl_1, \res\zl_2} \zs \foc{\tlf{P \llimp N}} \vdash \trp{\zat{z:\nat Q}}
        }{
          \infer[\fbox{\clr{\llimp}}]{
            \tlp{\unr\zl, \res\zl_1, \res\zl_2} \zs \foc{\trf{P} \llimp \tlf{N}} \vdash \trp{\zat{z:\nat Q}}
          }{
            \infer={
              \tlp{\unr\zl, \res\zl_1} \vdash \foc{\trf{P}} \zs \cdot
            }{
              \eic{\unr\zl, \res\zl_1 \vdash \foc{P}}
            }
            &
            \infer={
              \tlp{\unr\zl, \res\zl_2} \zs \foc{\tlf{N}} \vdash \trp{\zat{z:\nat Q}}
            }{
              \eic{\unr\zl, \res\zl_2 \zs \foc{N} \vdash \zat{z:\nat Q}}
            }
          }
        }
      }
    \end{kcgather}
    The boxed instance of $\clr{\llimp}$ contains the only split of
    the right context that can succeed in the same focused phase, \ie,
    reach an initial sequent or a phase transition, becaue that
    $\trf{P}$ eventually produces either a positive atom (which must
    finish the proof with $\rname{c/pr}$ and since right form zones
    are restricted $\trp{\zat{z:\nat Q}}$ cannot be present) or a
    $\llbang_{(z,\lj)}$ which guarantees that the rest of the right
    context is empty.

  \item \emph{cases of \ilr{\llqmark_z} and \rname{dr}}:
    \begin{kcgather}
      \infer={
        \eic{\zl \zs \foc{\llqmark_z\nat P} \vdash \zat{y:\nat Q}}
      }{
        \infer={
          \tlp{\zl} \zs \foc{\tlf{\llqmark_z\nat P}} \vdash \trp{\zat{y:\nat Q}}
        }{
          \infer[\fbox{\clr{\llqmark}}]{
            \tlp{\zl} \zs \foc{\llqmark_{(z,\rj)}\tla{\nat P}} \vdash \trp{\zat{y:\nat Q}}
          }{
            \infer={
              \tlp{\zl} \zs \tla{\nat P} \vdash \cdot \zs \trp{\zat{y:\nat Q}}
            }{
              \eic{\zl \zs \nat P \vdash \cdot \zs \zat{y:\nat Q}}
            }
          }
        }
      }
      \qquad
      \infer={
        \eic{\zl \zs \cdot \vdash \cdot \zs \zat{z:P}}
      }{
        \infer={
          \tlp{\zl} \zs \cdot \vdash \cdot \zs \trp{\zat{z:P}}
        }{
          \infer[\fbox{\rname{c/rdr}}]{
            \tlp{\zl} \zs \cdot \vdash \cdot \zs \zat{(z,\rj):\trf{P}}
          }{
            \infer={
              \tlp{\zl} \vdash \foc{\trf{P}} \zs \cdot
            }{
              \zl \vdash \foc{P}
            }
          }
        }
      }
    \end{kcgather}
    The boxed instance of \clr{\llqmark} is justified because the subscript zone $(z,\rj)$ is of the
    right form (in order to compare with $(y,\rj)$) which is $\hatle$-smaller than its corresponding
    left-form zone (\defnref{sig-split}). Note that it is crucial for soundness to have $(z, \rj)$
    not be smaller than all left form zones. Since right form zones are restricted, there is no
    copying in the boxed instance of \rname{c/rdr}. The other decision cases are similar. \qed
  \end{itemize}
\end{proof}

\noindent
We note one important direct corollary of \thmref{isl-csl}.

\begin{cor}[intuitionistic logic in classical linear logic] \label{thm:il-cll}
  There is a focally adequate encoding of intuitiontistic logic in
  classical linear logic.
\end{cor}

It is well known~\cite{girard87tcs} that (classical) linear logic can encode the intuitionistic
implication $\limp$ as follows: $A \limp B \triangleq \llbang A \llimp B$. However, this encoding is
only globally adequate~\cite{schellinx91jlc}. It is possible to refine this encoding to obtain a
fully adequate encoding~\cite{liang09lics-my} in an enriched classical linear logic which is not
apparently an instance of classical subexponential logic. Corollary~\ref{thm:il-cll} further
improves our undertanding of encodings of intuitionistic implicication by permuting $\llbang$ into
the antecedent of the implication until there is a phase change, which removes the bureaucratic
polarity switch inherent in this implication.\footnote{Note that the polarised intuitionistic
  implication $P \llimp N$, if encoded using Girard's encoding, would be $\llbang{{\uparrow}P}
  \llimp N$, which breaks the polarisation of the antecedent.}

\begin{proof}[of \corref{il-cll}]
  The split of the signature \texttt{l} (\defnref{csl-fam}) is
  isomorphic to the signature \texttt{ll}, so apply \thmref{isl-csl}.
  \qed
\end{proof}

\section{Conclusions}
\label{sec:concl}

Section \ref{sec:encodings} shows that any given classical (resp. intuitionistic) subexponential
logic can be encoded in a related intuitionistic (resp. classical) subexponential logic such that
partial synthetic derivations are preserved. This is a technical result, with at least one of the
directions of encoding being novel. It strongly suggests that one of the fractures in logic
identified by Miller in~\cite{miller10bcs-my}---the classical/intuitionistic divide---might be
healed by analyses and algorithms that are generic on subexponential signatures. One might still
favour ``classical'' or ``intuitionistic'' dialects for proofs, but neither format is more
fundamental.

The results of this paper have two caveats. First, we only consider the ``restricted'' or the
``unrestricted'' flavours of subexponentials; in~\cite{danos93kgc-my} there were also
subexponentials of the ``strict'' and ``affine'' flavours for which our results here do not extend
directly. Second, we do not consider encodings involving non-propositional kinds, such as terms or
frames. Subexponentials are still useful for such stronger encodings, but \emph{representational
  adequacy} may not be as straightforward.

\bibliographystyle{abbrv}
\bibliography{clasint,master}

\end{document}